\documentclass[twoside,11pt]{article}
\usepackage{jair, theapa, rawfonts}
\usepackage[utf8]{inputenc}
\usepackage{graphicx}
\usepackage[ruled,vlined]{algorithm2e}
\usepackage[table,xcdraw]{xcolor}
\usepackage{subfigure}
\usepackage{tcolorbox}
\usepackage{float}
\usepackage{tikz}
\usepackage{times}
\usepackage{amsmath}
\usepackage{amssymb}
\usepackage{amsthm}
\usetikzlibrary{decorations.pathreplacing}

\newtheorem{theorem}{Theorem}
\newtheorem{proposition}{Proposition}
\newtheorem{claim}{Claim}
\newtheorem{corollary}{Corollary}
\theoremstyle{definition}
\newtheorem{definition}{Definition}
\newtheorem{example}{Example}

\begin{document}

\title{Cost Sharing Public Project with Minimum Release Delay}
\author{\name Mingyu Guo \email mingyu.guo@adelaide.edu.au\\
       \name Diksha Goel \email diksha.goel@adelaide.edu.au\\
       \name Guanhua Wang \email guanhua.wang@adelaide.edu.au\\
       \name Yong Yang \email yong.yang@student.adelaide.edu.au\\
       \name Muhammad Ali Babar \email ali.babar@adelaide.edu.au\\
       \addr School of Computer Science, University of Adelaide, Australia}
\date{}
\maketitle

\begin{abstract}

    We study the excludable public project model where the decision is binary
    (build or not build).  In a classic excludable and binary public project
    model, an agent either consumes the project in its whole or is completely
    excluded. We study a setting where the mechanism can set different project
    release time for different agents, in the sense that high-paying agents can
    consume the project earlier than low-paying agents. The release delay,
    while hurting the social welfare, is implemented to incentivize payments to cover the project cost.
    The mechanism
    design objective is to minimize the {\em maximum release delay} and the {\em total release delay} among all agents.

    We first consider the setting where we know the prior distribution of the
    agents' types. Our objectives are minimizing the {\em expected} maximum
    release delay and the {\em expected} total release delay.
    We propose the {\em single deadline mechanisms}.  We show that the optimal single
    deadline mechanism is asymptotically optimal for both objectives,
    regardless of the prior distribution.  For small number of agents, we
    propose the {\em sequential unanimous mechanisms} by extending the {\em largest
    unanimous mechanisms} from Ohseto~\cite{Ohseto2000:Characterizations}. We
    propose an automated mechanism design approach via evolutionary computation
    to optimize within the sequential unanimous mechanisms.

    We next study prior-free mechanism design. We propose the {\em group-based
    optimal deadline mechanism} and show that it is {\em competitive} against an
    {\em undominated} mechanism under minor technical assumptions.


\end{abstract}

\section{Introduction}

The public project problem is a fundamental mechanism design model with many
applications in multiagent systems.  The public project problem involves
multiple agents, who need to decide whether or not to build a public project.
The project can be {\bf nonexcludable} ({\em i.e.}, if the project is built,
then every agent gets to consume the project, including the non-paying agents/free riders) or
{\bf excludable} ({\em i.e.}, the setting makes it possible to exclude some
agents from consuming the project)~\cite{Ohseto2000:Characterizations}.\footnote{An
example nonexcludable public project is a public airport, and an example
excludable public project is a gated swimming pool.}
A public project can be
{\bf indivisible/binary} or {\bf
divisible}~\cite{Moulin1994:Serial}.  A binary public project is either
built or not built ({\em i.e.}, there is only one level of provision).
In a divisible public project, there are multiple
levels of provision ({\em i.e.}, build a project with adjustable quality).

In this paper, we study an excludable public project model that is
``divisible'' in a different sense. In the model, the level of provision is
binary, but an agent's consumption is divisible. The mechanism specifies when
an agent can start consuming the project. High-paying agents can consume the
project earlier, and the free riders need to wait.  The waiting time is also
called an agent's {\bf delay}. The delay is there to incentivize payments.  The
model was motivated by the following cyber security information sharing scenario.
A group of agents come together to crowd fund a piece of security
information. For example, a practically happening scenario is that financial companies would purchase information regarding the latest ``hacks'' from security experts in order to prepare against them.\footnote{Price of high-value zero-day exploit can reach
millions~\cite{Greenberg2012:Shopping,Fisher2015:VUPEN}.} Based on the agents'
valuations on the information, the mechanism decides whether or not to
crowd fund this piece of information ({\em i.e.}, purchase it from the security
consulting firm that is selling this piece of information).  If we are able to
raise enough payments to cover the cost of the security information, then
ideally we would like to share it to all agents, including the free riders, in
order to maximizes the {\em overall protection of the community}.  However, if
all agents receive the information regardless of their payments, then no agents
are incentivized to pay. To address this, the mechanism releases the
information only to high-paying agents in the beginning and the
non-paying/low-paying agents need to wait for a delayed release.\footnote{Another example use case of our mechanism is to fund content creators.
For example, an author may ``kick start'' a new online novel series. If enough readers (the agents in this scenario) are willing to fund the novel by collectively covering the author's demanded cost, then the novel series will be written. Whenever a chapter is ready, those who pay will receive the new chapter right away. The free-riders will need to wait for a few days until the ``premium'' period is over.}
The mechanism
design goal is to minimize the delay as long as the delay is long enough to
incentivize enough payments to cover the cost of the information.
In this paper, we have two design objectives. One is to minimize the
{\em max-delay} ({\em i.e.}, the maximum waiting time of the agents) and the
other is to minimize the {\em sum-delay} ({\em i.e.}, the total waiting time of
the agents).


We first focus on settings where we know the prior distribution of the agents'
types.  We focus on minimizing the expected {\em max-delay} and the expected
{\em sum-delay}. We propose a mechanism family called the {\bf single deadline
mechanisms}.  For both objectives, under minor technical assumptions, we prove
that there exists a single deadline mechanism that is {\em near optimal} when
the number of agents is large, {\em regardless of the prior distribution}.
We show that when the number of agents approaches infinity, the optimal single
deadline mechanism approaches optimality asymptotically.  For small numbers of
agents, the single deadline mechanism is not optimal. We extend the single
deadline mechanisms to multiple deadline mechanisms. We also propose a genetic
algorithm based automated mechanism design approach.  We use a sequence of
offers to represent a mechanism and we evolve the sequences.  By simulating
mechanisms using multiple distributions, we show that our automated mechanism
design approach
successfully identifies well-performing mechanisms for small numbers of
agents.

We next focus on prior-free mechanism design. We propose the {\em optimal deadline
mechanism}, which we show is {\em undominated} by strategy-proof and
individually rational mechanisms. However, the optimal deadline mechanism
itself is not strategy-proof. We propose a strategy-proof variant of it, which we call the {\em group-based optimal deadline mechanism}. We show that the {\em group-based optimal deadline mechanism} is {\em competitive} (i.e., the delay is at most larger by a constant factor) against the undominated optimal deadline mechanism under minor technical assumptions.

\section{Related Research}

Ohseto~\cite{Ohseto2000:Characterizations} characterized
all strategy-proof and individually rational mechanisms for the binary public
project model (both excludable and nonexcludable), under minor technical
assumptions.
Deb and
Razzolini~\cite{Deb1999:Voluntary} further showed that on top of Ohseto's
characterization, if we require {\em equal treatment of equals} ({\em i.e.}, if
two agents have the same type, then they should be treated the same), then the
only strategy-proof and individually rational mechanisms are the {\em conservative equal cost mechanism} (nonexcludable) and the
{\em serial cost sharing mechanism} (excludable), which were both proposed by
Moulin~\cite{Moulin1994:Serial}. It should be noted that Ohseto's
characterization involves {\em exponential} number of parameters, so
knowing the characterization does not mean it is easy to locate
good mechanisms.  Wang {\em et al.}~\cite{Wang2021:Mechanism}
proposed a neural network based approach for optimizing within Ohseto's characterization
family.
The authors studied two objectives: maximizing the
 number of consumers and maximizing the social welfare.
 It should be noted that Ohseto's characterization does not apply to the model
 in this paper, as our model has an additional spin that is the release delay.
 In this paper, for minimizing the expected delay, we propose a family of mechanisms called the sequential
 unanimous mechanisms, which is motivated by Ohseto's characterization.
 We apply a genetic algorithm for tuning the sequential unanimous mechanisms.
Mechanism
design via evolutionary computation~\cite{Phelps2010:Evolutionary} and
mechanism design via other computational means (such as linear
programming~\cite{Conitzer2002:Complexity} and neural
networks~\cite{Duetting2019:Optimal,Shen2019:Automated,Wang2021:Mechanism})
have long been shown to be effective for many design settings.



\section{Model Description}\label{sec:model}

There are $n$ agents who decide whether or not to build a public project. The
project is binary (build or not build) and nonrivalrous (the cost of the
project does not depend on how many agents are consuming it). We normalize the
project cost to $1$. Agent $i$'s type $v_i \in [0, 1]$ represents her private
valuation for the project. We use $\vec{v}=(v_1,v_2,\ldots,v_n)$ to denote the type profile.
In Section~\ref{sec:expected}, we focus on mechanism design assuming a known
prior distribution of the agents' types. Specifically, we assume that the $v_i$ are drawn {\em
i.i.d.}, with $f$ being the probability
density function.  For technical reasons, we assume
$f$ is {\em positive} and {\em Lipschitz continuous over $[0,1]$}.

We assume that the public project has value over a time period $[0,1]$.
For example, the project could be a piece of security information that is discovered
at time $0$ and the corresponding exploit expires at time $1$.  We
assume the setting allows the mechanism to specify each agent's release time for the project, so that some agents can consume the project earlier than the others.  Given a type
profile, a mechanism outcome consists of two vectors: $(t_1,t_2,\ldots,t_n)$
and $(p_1,p_2,\ldots,p_n)$. {\em I.e.}, agent $i$ starts consuming the project
at time $t_i \in [0,1]$ and pays $p_i\ge 0$.  $t_i=0$ means agent $i$ gets to consume
the public project right from the beginning and $t_i=1$ means agent $i$ does not
get to consume the public project.
We call $t_i$ agent $i$'s {\em release time}.
We assume the agents' valuations over the time
period is uniform. That is, agent $i$'s valuation equals $v_i(1-t_i)$, as she
enjoys the time interval $[t_i,1]$, which has length $1-t_i$. Agent $i$'s utility
is then $v_i(1-t_i)-p_i$.
We impose the following mechanism design constraints:

\begin{itemize}

    \item Strategy-proofness: We use $t_i$ and $p_i$ to denote agent $i$'s release
        time and payment when she reports her true value $v_i$. We use $t_i'$ and $p_i'$ to denote agent $i$'s release
        time and payment when she reports a false value $v_i'$. We should have
        \[v_i(1-t_i)-p_i \ge v_i(1-t_i')-p_i'\]

    \item Individual rationality: $v_i(1-t_i)-p_i\ge 0$

    \item Ex post budget balance:

        {\em If the project is not built}, then no agent
        can consume the project and no agent pays. That is, we must have $t_i=1$ and $p_i=0$ for all $i$.

        {\em If the project is built}, then the agents' total payment must cover
        exactly the project cost. That is, $\sum_i p_i=1$.

\end{itemize}

In Section~\ref{sec:expected} and \ref{sec:prior-free}, we study different
mechanism design objectives. We defer the formal definitions of our mechanism design objectives
to the specific sections.

\section{Minimizing Expected Max-Delay and Sum-Delay}
\label{sec:expected}

In this section, we design mechanisms that minimize the following design objectives:

\begin{itemize}
    \item Expected {\em Max-Delay}: $E_{v_i\sim f}\left(\max\{t_1,t_2,\ldots,t_n \}\right)$
    \item Expected {\em Sum-Delay}: $E_{v_i\sim f}\left(\sum_{i}t_i\right)$
\end{itemize}

\subsection{Single Deadline Mechanisms}

We first describe the {\em serial cost sharing mechanism (SCS)} proposed by
Moulin~\cite{Moulin1994:Serial}.  The SCS mechanism was proposed for the
classic binary excludable model, so the concept of release time does not exist
in its original definition.  In the context of our model, under SCS, an agent's
release time is either $0$ or $1$. ($0$ means the agent gets to consume
the project from the beginning in its whole, and $1$ means that the agent does
not get to consume the project.)

Let $\vec{v}$ be the type profile.
We first define the following functions:

    \[I(\vec{v})=\begin{cases}
        1 & \exists k \in \{1,2,\ldots,n\}, k \le |\{v_i|v_i\ge \frac{1}{k}\}| \\
        0 & \text{otherwise}
   \end{cases}
    \]

$I(\vec{v})$ equals $1$ if and only if there exist at least $k$ values among $\vec{v}$
that are at least $\frac{1}{k}$, where $k$ is an integer from $1$ to $n$.

    \[K(\vec{v})=\begin{cases}
        \max\{k| k \le |\{v_i|v_i\ge \frac{1}{k}\}|,k\in \{1,2,\ldots,n\}\} & I(\vec{v})=1\\
        0 & I(\vec{v})=0
   \end{cases}
    \]

Given $\vec{v}$, there could be multiple values for $k$, where there exist at
least $k$ values among $\vec{v}$ that are at least $\frac{1}{k}$. $K(\vec{v})$
is the largest value for $k$. If such a $k$ value does not exist, then
$K(\vec{v})$ is set to $0$.

\begin{definition}[Serial Cost Sharing Mechanism~\cite{Moulin1994:Serial}]

    Given $\vec{v}$, let $k=K(\vec{v})$.

    \begin{itemize}

        \item
    If $k>0$, then agents with the highest $k$ values are the consumers.
            The consumers pay $\frac{1}{k}$. The non-consumers do not pay.

        \item
    If $k=0$, then there are no consumers and no agents pay.

    \end{itemize}

\end{definition}

In the context of our model, consumers' release time/delay is $0$
and non-consumers' release time/delay is $1$.
Essentially, the serial cost sharing mechanism finds the largest $k$ where $k$
    agents are willing to equally split the cost. If such a $k$ exists, then we say {\em the cost share is successful} and these $k$ agents are {\em joining the cost share}.
    If such a $k$ does not exist, then we say {\em the cost share failed}.

Next we introduce a new mechanism family called the single deadline mechanisms, which
use the serial cost sharing mechanism as a sub-component.

\begin{definition}[Single Deadline Mechanisms]

A single deadline mechanism is characterized by one parameter $d \in [0,1]$.
    $d$ is called the mechanism's {\bf deadline}.
    We use $M(d)$ to denote the single deadline mechanism with deadline $d$.

    The time interval before the deadline $[0,d]$ is called the {\bf non-free}
    part. The time interval after the deadline $[d,1]$ is called the {\bf free}
    part.

    We run the serial cost sharing mechanism on the non-free part as follows.  For the
    non-free part, the agents' valuations are
    $d\vec{v}=(dv_1,\ldots,dv_n)$.
    Let $k=K(d\vec{v})$.
    Agents with the highest $k$ values get to consume the non-free part, and
    each needs to pay $\frac{1}{k}$.

    The free part is allocated to the agents for free. However, we cannot
    give out the free part if the public project is not built.

    If we give out the free part if and only if $I(d\vec{v})=1$,
    then the mechanism is not strategy-proof, because the free parts change the
    agents' strategies.\footnote{For example, an agent may over-report to turn an
    unsuccessful cost share into a successful cost share, in order to claim the
    free part.}
    Instead, we give agent $i$ her free part if and only if $I(dv_{-i})=1$.
    That is, agent $i$ gets her free part if and only if the other agents
    can successfully cost share the non-free part without $i$.

    If an agent receives both the non-free part and the free part, then her release
    time is $0$. If an agent only receives the free part, then her release time
    is $d$. If an agent does not receive either part, then her release time is $1$.
    Lastly, if an agent only receives the non-free part, then her release time
    is $1-d$, because such an agent's consumption interval should have length $d$ ({\em i.e.}, $[1-d,1]$).
\end{definition}

\begin{proposition}
    The single deadline mechanisms are strategy-proof, individually rational, and
    ex post budget balanced.
\end{proposition}

\begin{proof}
    Whether an agent receives her free part or not does not depend on her report,
    so the agents are essentially just facing a serial cost sharing mechanism, where
    the item being cost shared is $d$ portion of the public project.
    The serial cost
    sharing mechanism is strategy-proof and individually rational.
    If $I(d\vec{v})=0$, then an agent receives neither the free nor the non-free part.
    Every agent's release time equals $1$ and every agent pays $0$.
    If $I(d\vec{v})=1$, then the total payment is exactly $1$. Therefore, the
    single deadline mechanism is ex post budget balanced.
\end{proof}

We next show that the above single deadline mechanism family contains
the asymptotically
optimal mechanism for both the {\em max-delay} and the {\em sum-delay}
objectives, regardless of the prior distribution.

\subsection{Max-Delay: Asymptotic Optimality}


\begin{theorem}\label{thm:maxdelay}
    The optimal single deadline mechanism's expected max-delay approaches $0$
    when the number of agents approaches infinity.
\end{theorem}

\begin{proof}
We consider a single deadline mechanism $M(d)$.  Every agent's valuation is drawn {\em i.i.d.} from a distribution with PDF $f$.
Let $V_i$ be the random variable representing agent $i$'s valuation.
Since $f$ is positive and Lipschitz continuous, we have that
$\forall d, \exists k, P(dV_i\ge \frac{1}{k}) >0$.
That is, for any deadline $d$, there always exists an integer $k$, where the
probability that an agent is willing to pay $\frac{1}{k}$ for the non-free part is positive.
Let $p=P(dV_i\ge \frac{1}{k})$.
We define the following Bernoulli random variable:
    \[B_i=\begin{cases}
        1 & dV_i\ge \frac{1}{k}\\
        0 & \text{otherwise}
   \end{cases}
    \]

$B_i$ equals $1$ with probability $p$. It equals $1$ if and only if agent $i$
can afford $\frac{1}{k}$ for the non-free part.
The total number of agents in $\vec{v}$ who can afford
$\frac{1}{k}$ for the non-free part then follows a Binomial distribution
$B(n,p)$. We use $B$ to denote this Binomial variable.
If $B\ge k+1$, then every agent receives the free part, because
agent $i$ receives the free part if excluding herself, there are at least $k$
agents who are willing to pay $\frac{1}{k}$ for the non-free part.
The probability that the max-delay is higher than $d$ is therefore bounded above
by $P(B\le k)$.
According to Hoeffding's inequality, when $k<np$,
$P(B\le k)\le e^{-2n\left(p-\frac{k}{n}\right)^2}$.
We immediately have that when $n$ approaches infinity, the probability
that the max-delay is higher than $d$ is approaching $0$. Since $d$
is arbitrary, we have that asymptotically, the single deadline mechanism's
expected max-delay is approaching $0$.
\end{proof}

Next, we use an example to show that when $n=500$,
the optimal single deadline mechanism's expected max-delay is close to $0.01$.
We reuse all notation defined in the proof of Theorem~\ref{thm:maxdelay}.
We make use of the Chernoff bound. When $k<np$, we have
$P(B\le k)\le e^{-nD\left(\frac{k}{n} ||p \right)}$, where $D\left(a||p\right)=a\ln\frac{a}{p}+(1-a)\ln\frac{1-a}{1-p}$.

When all agents receive the free part, the max-delay is at most $d$.
Otherwise, the max-delay is at most $1$.
The expected max-delay is at most
\[P(B\le k) + d(1 - P(B\le k)) \le P(B\le k)+d\]

\begin{example}\label{ex:max}
    Let us consider a case where $n=500$. We set $d=0.01$ and $k=250$.

    \begin{itemize}

        \item $f$ is the uniform distribution $U(0,1)$:
            We have $p=0.6$ and $P(B\le 250)\le 3.69\mathrm{e}-5$.
            $M(0.01)$'s expected max-delay is then bounded above by $0.01 + 3.69\mathrm{e}-5$.

        \item $f$ is the normal distribution $N(0.5,0.1)$ restricted to $[0,1]$:
            We have $p=0.84$ and $P(B\le 250)\le 7.45\mathrm{e}-69$.
            $M(0.01)$'s expected max-delay is then bounded above by $0.01 + 7.45\mathrm{e}-69$.

    \end{itemize}

\end{example}

On the contrary, the expected max-delay of the serial cost sharing mechanism is not approaching
$0$ asymptotically.

\begin{proposition}
    The expected max-delay of the serial cost sharing mechanism equals
    \[1-(\int_{\frac{1}{n}}^1f(x)dx)^n\]

    The above expression approaches $1-e^{-f(0)}$ asymptotically.
\end{proposition}

\begin{proof}
    The max-delay is $1$ whenever at least one agent's type is less than $\frac{1}{n}$.
    Under our assumption that $f$ is Lipschitz continuous, the limit of the
    above expression can be calculated via the L'Hospital's rule.
\end{proof}

For example, when $n=500$, under $U(0,1)$,
the expected max-delay of the serial cost sharing mechanism equals $0.632$,
which is very close to $1-\frac{1}{e}$.

\subsection{Sum-Delay: Asymptotic Optimality}


\begin{theorem}\label{thm:sumdelay}
    When the number of agents approaches infinity, the optimal single deadline mechanism is optimal among all mechanisms in terms of expected sum-delay.
\end{theorem}

Theorem~\ref{thm:sumdelay} can be proved by combining
Proposition~\ref{prop:lowerboundsum} and
Proposition~\ref{prop:achievesumdelay}.

\begin{proposition}\label{prop:finite}
    The optimal expected sum-delay is finite regardless of the distribution.
\end{proposition}

\begin{proof}
    We consider the following mechanism: Pick an arbitrary integer $k>1$. We offer $\frac{1}{k}$ to the agents one by one. An agent gets the whole interval $[0,1]$ if she agrees
    to pay $\frac{1}{k}$ and if the project is built. Otherwise, she gets nothing. We build the project only when $k$ agents agree. Since we approach the agents one by one, after $k$ agents agree to pay $\frac{1}{k}$, all future agents receive the whole interval for free. This mechanism's expected sum-delay is bounded above by a constant. The constant
    only depends on the value of $k$ and the distribution.
\end{proof}

The following proposition follows from Proposition~\ref{prop:finite}.

\begin{proposition}\label{prop:fail}
    Given a mechanism $M$ and the number of agents $n$, let $Fail(n)$ be the probability of not building under $M$. We only need to consider $M$ that satisfies $Fail(n)= O(1/n)$.
\end{proposition}




We then propose a relaxed version of the ex post budget balance constraint, and
use it to calculate the delay lower bound.

\begin{definition}[Ex ante budget balance]
    Mechanism $M$ is ex ante budget balanced if and only if the expected total payment from the agents equals the probability of building, which is
    equal to the expected project cost (as the project cost is exactly $1$).
\end{definition}

\begin{proposition}\label{prop:lowerboundsum}
    Let $Fail(n)$ be the probability of not building the project when there are $n$ agents.
    We consider what happens when we offer $o$ for the whole interval $[0,1]$ to an individual
    agent. If the agent accepts $o$ then she pays $o$ and gets the whole interval. Otherwise,
    the agent pays $0$ and receives nothing.

    We define the {\em delay versus payment ratio} $r(o)$ as follows:

    \[r(o) = \frac{\int_{0}^{o}f(x)dx}{o\int_o^1f(x)dx}\]

    $r$ is continuous on $(0,1)$. Due to $f$ being Lipschitz continuous, we have $\lim_{o\rightarrow 0}r(o)=f(0)$
    and $\lim_{o\rightarrow 1}r(o)=\infty$.\footnote{When $o$ approaches $0$,
    $r(o)$'s numerator is approaching $of(0)$ while the denominator is approaching
    $o$.}
    We could simply set $r(0)=f(0)$, then $r$ is continuous on $[0,1)$.
    We define the {\em optimal delay versus payment ratio}
    $r^* = \min_{o\in [0,1)}r(o)$.


    The expected sum-delay is bounded below by $r^*(1-Fail(n))$, which approaches
    $r^*$ asymptotically according to Proposition~\ref{prop:fail}.



\end{proposition}

\begin{proof}
    If we switch to ex ante budget balance, then it is without loss of generality to
    focus on anonymous mechanisms. We then face a single agent mechanism design problem
    where an agent pays $\frac{1-Fail(n)}{n}$ in expectation and we want to minimize
    her expected delay. Based on Myerson's characterization for single-parameter settings, here every strategy-proof mechanism works as follows: for each infinitesimal time interval
    there is a price and the price increases as an agent's allocated interval increases in length. There is an optimal price that minimizes the ratio between the delay caused by the price and the expected payment (price times the probability that the price is accepted).
    The total payment is $1-Fail(n)$, which means the total delay is at least $r^*(1-Fail(n))$.
\end{proof}

\begin{proposition}\label{prop:achievesumdelay}
    Let $o^*$ be the optimal offer that leads to the optimal delay versus payment ratio
    $r^*$.\footnote{If $o^*=0$, then we replace it with an infinitesimally small $\gamma>0$. The achieved
    sum-delay is then approaching $r(\gamma)(1+\epsilon)$ asymptotically. When $\gamma$ approaches
    $0$, $r(\gamma)$ approaches $r^*$.}
    \[o^* = \arg\min_{o\in [0,1)} r(o)\]

    Let $\epsilon >0$ be an arbitrarily small constant.
    The following single deadline mechanism's expected sum delay approaches $r^*(1+\epsilon)$ asymptotically.

    \[M(\frac{1+\epsilon}{no^*\int_{o^*}^1f(x)dx})\]

\end{proposition}

\begin{proof}
Let $p=P(V_i(1+\epsilon)\ge o^*)$.
Let $k=n\int_{o^*}^1f(x)dx$.
$p$ is the probability that an agent is willing to pay $\frac{1}{k}$
for the non-free part whose length is $\frac{1+\epsilon}{no^*\int_{o^*}^1f(x)dx}$.
We use $B$ to denote the Binomial distribution $B(n,p)$.
If $B> k$, then every agent receives the free part, because
cost sharing is successful even if we remove one agent.
The probability that an agent does not receive the free part is then bounded above
by $P(B\le k)$.
According to Hoeffding's inequality, we have that when $k<np$, we have

\[P(B\le k)\le e^{-2n(p-\frac{k}{n})^2} =
e^{-2n(\int_{\frac{o^*}{1+\epsilon}}^1f(x)dx-\int_{o^*}^1f(x)dx)^2}
=e^{-2n(\int_{\frac{o^*}{1+\epsilon}}^{o^*}f(x)dx)^2}\]

Let $\beta=\int_{\frac{o^*}{1+\epsilon}}^{o^*}f(x)dx$.
The expected total delay when some agents do not receive the free part is then at most $ne^{-2n\beta^2}$, which approaches $0$ as $n$ goes to infinity.
Therefore, we only need to consider situations where all agents receive the free part and at least $k$ agents receive the non-free part.
The expected sum delay on the remaining $n-k$ agents is then at most

    \[(n-k)\frac{1+\epsilon}{no^*\int_{o^*}^1f(x)dx}= \frac{(\int_0^{o^*}f(x)dx)(1+\epsilon)}{o^*\int_{o^*}^1f(x)dx}= (1+\epsilon)r^*\]

This concludes the proof.
\end{proof}


We then use an example to show that when $n=500$, under different distributions,
the optimal single deadline mechanism's expected sum-delay is close to optimality.

\begin{example}
    We consider $n=500$ which is the same as Example~\ref{ex:max}.
    Simulations are based on $100,000$ random draws.

    \begin{itemize}

        \item $f$ is the uniform distribution $U(0,1)$:
            The single deadline mechanism $M(1)$ (essentially the serial cost sharing mechanism)
            has an expected sum-delay of $1.006$, which is calculated via numerical simulation.
            $Fail(500)$ is then at most $0.002$.
            $r^*=1$. The lower bound is $0.998$, which
            is close to our achieved sum-delay $1.006$.

        \item $f$ is the normal distribution $N(0.5,0.1)$ restricted to $[0,1]$:
            The single deadline mechanism $M(1)$'s expected sum-delay equals $2.3\mathrm{e}-4$ in simulation, which is obviously close to optimality.

        \item $f$ is the beta distribution $Beta(0.5,0.5)$:
            The single deadline mechanism $M(0.01)$'s expected sum-delay equals $1.935$ in simulation.
            $Fail(500)$ is then at most $0.00387$.
            $r^*=1.927$. The lower bound equals $(1-0.00387)*r^*=1.920$, which
            is very close to the achieved sum-delay of $1.935$.
            The serial cost sharing mechanism $M(1)$ is far away from
            optimality in this example.
            The expected sum-delay of the serial cost sharing mechanism
            is much larger at $14.48$.
    \end{itemize}

\end{example}







\subsection{Automated Mechanism Design for Smaller Numbers of Agents}

For smaller numbers of agents,
the single deadline mechanism family no longer contains a near optimal mechanism.
We instead propose two numerical methods for identifying better mechanisms.
One is by extending the single deadline mechanism family
and the other is via evolutionary computation.


\begin{definition}[Multiple Deadline Mechanisms]
    A multiple deadline mechanism\\
    $M(d_1,\ldots,d_n)$ is characterized
    by $n$ different deadlines. Agent $i$'s non-free part is $[0,d_i]$ and
    her free part is $[d_i,1]$. The mechanism's rules are otherwise identical
    to the single deadline mechanisms.
\end{definition}

We simply use exhaustive search to find the best set of deadlines. Obviously,
this approach only works when the number of agents is tiny.
We then present an Automated Mechanism Design approach based on evolutionary computation.

Ohseto~\cite{Ohseto2000:Characterizations} characterized all strategy-proof and
individually rational mechanisms for the binary public project model (under
several minor technical assumptions). We summarize the author's characterization as follows:

\begin{itemize}

    \item {\em Unanimous mechanisms} (characterization for the nonexcludable model): Under an unanimous mechanism, there is a cost share vector $(c_1,c_2,\ldots,c_n)$ with $c_i\ge 0$ and
        $\sum_{i}c_i=1$. The project is built if and only if all agents accept this cost share vector.

    \item {\em Largest unanimous mechanisms} (characterization for the
        excludable model): Under a largest unanimous
        mechanism, for every subset/coalition of the agents, there is a constant cost
        share vector.
        The agents initially face the cost share vector corresponding to the grand coalition. If some
        agents do not accept the current cost share vector, then they are forever excluded. The remaining
        agents face a different cost share vector based on who are left.
        If at some point, all remaining agents accept,
        then we build the project. Otherwise, the project is not built.
\end{itemize}

We extend the largest unanimous mechanisms by adding the {\em release time} element.

\begin{definition}[Sequential unanimous mechanisms]
    A cost share vector under a sequential unanimous mechanism includes
    both the payments and the release time:

    \[T_1,B_1,\quad T_2,B_2,\quad \ldots,\quad T_n,B_n\]

    Agent $i$ accepts the above cost share vector if and only if her utility
    {\em based on her reported valuation} is nonnegative when paying $B_i$
    for the time interval $[T_i,1]$.  That is, agent $i$ accepts the above cost share
    vector if and only if {\em her reported valuation} is at least
    $\frac{B_i}{1-T_i}$.  $\frac{B_i}{1-T_i}$ is called the {\em
    unit price} agent $i$ faces. We require $B_i\ge 0$ and $\sum_{i}B_i=1$.

    A sequential unanimous mechanism contains $m$ cost share vectors in a
    sequence.  The mechanism goes through the sequence and stops at the first
    vector that is accepted by all agents. The project is built and the agents'
    release time and payments are determined by the unanimously accepted cost
    share vector. If all cost share vectors in the sequence are rejected, then
    the decision is not to build.
\end{definition}

The largest unanimous mechanisms (can be interpreted as special cases with
binary $T_i$) form a subset of the sequential unanimous mechanisms.  The
sequential unanimous mechanisms' structure makes it suitable for genetic
algorithms --- we treat the cost share vectors as the {\em genes} and treat the
sequences of cost share vectors as the {\em gene sequences}.

The sequential unanimous mechanisms are generally not strategy-proof. However, they
can be easily proved to be strategy-proof in two scenarios:

\begin{itemize}

    \item A sequential unanimous mechanism is strategy-proof when {\em the sequence
        contains only one cost share vector} (an agent faces a take-it-or-leave-it
        offer).
        This observation makes it easy to generate
        an initial population of strategy-proof mechanisms.

    \item If for every agent, as we go through the cost share vector sequence,
        the unit price an agent faces is {\em nondecreasing} and her release time is also
        {\em nondecreasing}, then the mechanism is strategy-proof.
        Essentially, when the above is satisfied,
        all agents prefer earlier cost share vectors.
        All agents are incentivized to report truthfully, as doing so enables them to secure the earliest possible cost share vector.
\end{itemize}

The sequential unanimous mechanism family
{\em seems} to be quite expressive.\footnote{Let $M$ be a strategy-proof
mechanism. There exists a sequential unanimous mechanism $M'$ (with exponential
sequence length). $M'$ has an
approximate equilibrium where the equilibrium outcome is
arbitrarily close to $M$'s outcome. To prove this, we only need to discretize an
individual agent's
type space $[0,1]$ into a finite number of grid points. The number of type profiles
is exponential. We place $M$'s outcomes for all these type profiles
in a sequence.} Our experiments show that by optimizing
within the sequential unanimous mechanisms, we are able to identify mechanisms
that perform better than existing mechanisms.
Our approach is as follows:

\begin{itemize}

    \item Initial population contains $200$ strategy-proof mechanisms. Every initial mechanism is
        a sequential unanimous mechanism with only one cost share vector.
        The $B_i$ and the $T_i$ are randomly generated by sampling $U(0,1)$.

    \item We perform evolution for $200$ rounds. Before each round, we filter
        out mechanisms that are not truthful. We have two different filters:

        \begin{itemize}

            \item Strict filter: we enforce that every agent's unit price faced and
                release time must be nondecreasing. With this filter, the final
                mechanism produced must be strategy-proof. We call this variant
                the {\em Truthful Genetic Algorithm (TGA)}.

            \item Loose filter: we use simulation to check for
                strategy-proofness violations.  In every evolution round, we
                generate $200$ random type profiles.  For each type profile and
                each agent, we randomly draw one false report and we filter out
                a mechanism if any beneficial manipulation occurs.
                After finishing evolution, we use $10,000$ type profiles to
                filter out the untruthful mechanisms from the final population.
                It should be noted that, we can only claim that the
                remaining mechanisms are {\em probably} truthful.
                We call this variant
                the {\em Approximately Truthful Genetic Algorithm (ATGA)}.

        \end{itemize}

    \item We perform crossover and mutations as follows:

        \begin{itemize}

            \item Crossover: We call the top $50\%$ of the
                population (in terms of fitness, {\em i.e.}, expected max-delay or sum-delay)
                the {\em elite population}. For every elite
                mechanism, we randomly pick another mechanism from the whole
                population, and perform a crossover by randomly swapping one gene segment.

\begin{figure}[h!]
\includegraphics[width=\textwidth]{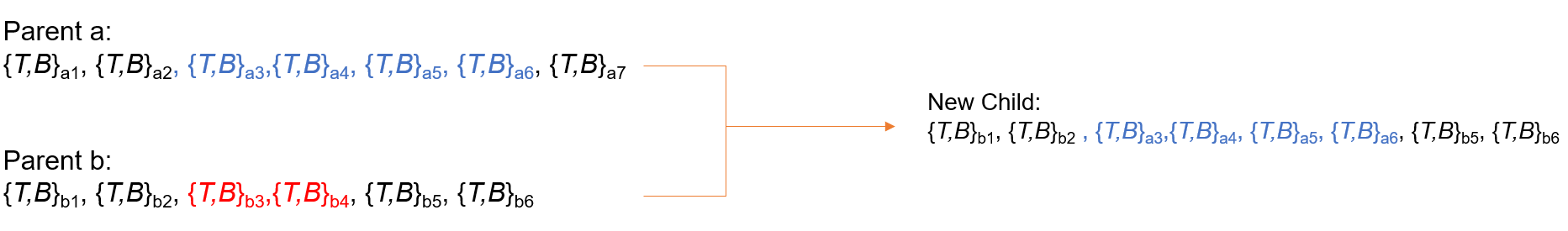}
\caption{Crossover}\label{fig1}
\end{figure}

            \item Mutation: For every elite mechanism, with $20\%$ chance, we randomly select
                one gene, modify the offer of one agent by making it worse.
                We insert that new cost share vector
                into a random position after the original position.

\begin{figure}[h!]
\includegraphics[width=\textwidth]{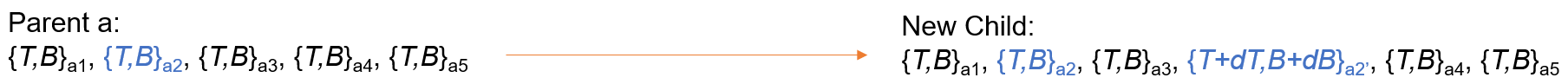}
\caption{Mutation}\label{fig2}
\end{figure}

            \item Neighbourhood Search: For every elite mechanism, with $20\%$ chance,
                we randomly perturb one gene uniformly (from $-10\%$ to $+10\%$).

\begin{figure}[h!]
\includegraphics[width=\textwidth]{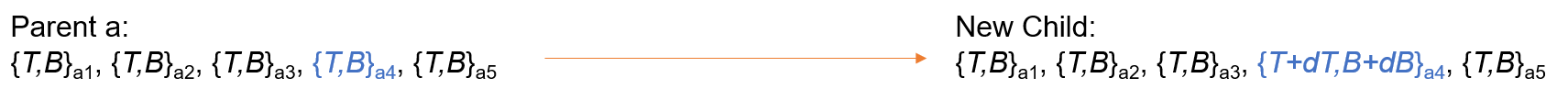}
\caption{Neighborhood Search}\label{fig3}
\end{figure}

        \item Abandon duplication and unused genes: In every evolution round,
    if a cost share vector is never unanimously accepted or if two cost share vectors
    are within $0.0001$ in L1 distance.
    then we remove the duplication/unused genes.

        \end{itemize}

\end{itemize}










We present the expected max-delay and sum-delay for $n=3,5$ and for different
distributions.  ATGA is only approximately truthful. We recall that in our
evolutionary process, in each round, we only use a very loose filter to filter
out the untruthful mechanisms. After evolution finishes, we run a more rigorous
filter on the final population (based on $10,000$ randomly generated type
profiles). The percentage in the parenthesis is the percentage of mechanisms
surviving the more rigorous test.  The other mechanisms (TGA and Single/Multiple
deadlines) are strategy-proof. SCS is the serial cost sharing mechanism from
Moulin~\cite{Moulin1994:Serial}.



\begin{table}[]
\centering{}
\small
\begin{tabular}{cccccc}
\multicolumn{1}{c|}{\textit{\textbf{n=3,sum-delay}}} & {\color[HTML]{303030} ATGA}                  & {\color[HTML]{303030} TGA}            & {\color[HTML]{303030} Single deadline} & {\color[HTML]{303030} Multiple deadline} & {\color[HTML]{303030} SCS}            \\ \hline
\multicolumn{1}{c|}{\textbf{Uniform(0,1)}}           & {\color[HTML]{303030} \textbf{1.605(95\%)}}  & {\color[HTML]{303030} \textbf{1.605}} & {\color[HTML]{303030} \textbf{1.605}}  & {\color[HTML]{303030} \textbf{1.605}}    & {\color[HTML]{303030} \textbf{1.605}} \\
\multicolumn{1}{c|}{\textbf{Beta(0.5,0.5)}}          & {\color[HTML]{303030} \textbf{1.756(89\%)}}  & {\color[HTML]{303030} \textbf{1.757}} & {\color[HTML]{303030} \textbf{1.757}}  & {\color[HTML]{303030} \textbf{1.757}}    & {\color[HTML]{303030} \textbf{1.757}} \\
\multicolumn{1}{c|}{\textbf{Bernoulli(0.5)}}         & {\color[HTML]{303030} \textbf{0.869(100\%)}} & {\color[HTML]{303030} \textbf{0.868}} & {\color[HTML]{303030} 1.499}           & {\color[HTML]{303030} 1.253}             & {\color[HTML]{303030} 1.498}          \\
\multicolumn{1}{c|}{\textbf{50\% 0, 50\% 0.8}}       & {\color[HTML]{303030} \textbf{1.699(98\%)}}  & {\color[HTML]{303030} 1.873}          & {\color[HTML]{303030} 1.873}           & {\color[HTML]{303030} 1.873}             & {\color[HTML]{303030} 1.873}          \\
\textbf{}                                            & {\color[HTML]{303030} \textbf{}}             & {\color[HTML]{303030} }               & {\color[HTML]{303030} }                & {\color[HTML]{303030} }                  & {\color[HTML]{303030} }               \\
\multicolumn{1}{c|}{\textit{\textbf{n=3,max-delay}}} & {\color[HTML]{303030} ATGA}                  & {\color[HTML]{303030} TGA}            & {\color[HTML]{303030} Single deadline} & {\color[HTML]{303030} Multiple deadline} & {\color[HTML]{303030} SCS}            \\ \hline
\multicolumn{1}{c|}{\textbf{Uniform(0,1)}}           & {\color[HTML]{303030} \textbf{0.705(97\%)}}  & {\color[HTML]{303030} \textbf{0.705}} & {\color[HTML]{303030} \textbf{0.705}}  & {\color[HTML]{303030} \textbf{0.705}}    & {\color[HTML]{303030} \textbf{0.705}} \\
\multicolumn{1}{c|}{\textbf{Beta(0.5,0.5)}}          & {\color[HTML]{303030} \textbf{0.754(87\%)}}  & {\color[HTML]{303030} 0.757}          & {\color[HTML]{303030} 0.782}           & {\color[HTML]{303030} 0.757}             & {\color[HTML]{303030} 0.782}          \\
\multicolumn{1}{c|}{\textbf{Bernoulli(0.5)}}         & {\color[HTML]{303030} \textbf{0.5(100\%)}}   & {\color[HTML]{303030} \textbf{0.498}} & {\color[HTML]{303030} 0.687}           & {\color[HTML]{303030} \textbf{0.50}}     & {\color[HTML]{303030} 0.877}          \\
\multicolumn{1}{c|}{\textbf{50\% 0, 50\% 0.8}}       & {\color[HTML]{303030} \textbf{0.676(94\%)}}  & {\color[HTML]{303030} 0.753}          & {\color[HTML]{303030} 0.749}           & {\color[HTML]{303030} 0.749}             & {\color[HTML]{303030} 0.877}          \\
\multicolumn{1}{l}{}                                 & \multicolumn{1}{l}{}                         & \multicolumn{1}{l}{}                  & \multicolumn{1}{l}{}                   & \multicolumn{1}{l}{}                     & \multicolumn{1}{l}{}                  \\
\multicolumn{1}{c|}{\textit{\textbf{n=5,sum-delay}}} & {\color[HTML]{303030} ATGA}                  & {\color[HTML]{303030} TGA}            & {\color[HTML]{303030} Single deadline} & {\color[HTML]{303030} Multiple deadline} & {\color[HTML]{303030} SCS}            \\ \hline
\multicolumn{1}{c|}{\textbf{Uniform(0,1)}}           & {\color[HTML]{303030} 1.462(95\%)}           & {\color[HTML]{303030} 1.503}          & {\color[HTML]{303030} \textbf{1.415}}  & {\color[HTML]{303030} \textbf{1.415}}    & {\color[HTML]{303030} \textbf{1.415}} \\
\multicolumn{1}{c|}{\textbf{Beta(0.5,0.5)}}          & {\color[HTML]{303030} 2.279(92\%)}           & {\color[HTML]{303030} 2.12}           & {\color[HTML]{303030} \textbf{1.955}}  & {\color[HTML]{303030} \textbf{1.955}}    & {\color[HTML]{303030} \textbf{1.955}} \\
\multicolumn{1}{c|}{\textbf{Bernoulli(0.5)}}         & {\color[HTML]{303030} \textbf{1.146(100\%)}} & {\color[HTML]{303030} 1.867}          & {\color[HTML]{303030} 2.106}           & {\color[HTML]{303030} 1.711}             & {\color[HTML]{303030} 2.523}          \\
\multicolumn{1}{c|}{\textbf{50\% 0, 50\% 0.8}}       & {\color[HTML]{303030} 2.432(94\%)}           & {\color[HTML]{303030} 2.845}          & {\color[HTML]{303030} 2.323}           & {\color[HTML]{303030} \textbf{2.248}}    & {\color[HTML]{303030} 2.667}          \\
\textbf{}                                            & {\color[HTML]{303030} \textbf{}}             & {\color[HTML]{303030} }               & {\color[HTML]{303030} }                & {\color[HTML]{303030} }                  & {\color[HTML]{303030} }               \\
\multicolumn{1}{c|}{\textit{\textbf{n=5,max-delay}}} & {\color[HTML]{303030} ATGA}                  & {\color[HTML]{303030} TGA}            & {\color[HTML]{303030} Single deadline} & {\color[HTML]{303030} Multiple deadline} & {\color[HTML]{303030} SCS}            \\ \hline
\multicolumn{1}{c|}{\textbf{Uniform(0,1)}}           & {\color[HTML]{303030} 0.677(91\%)}           & {\color[HTML]{303030} 0.677}          & {\color[HTML]{303030} \textbf{0.662}}  & {\color[HTML]{303030} \textbf{0.662}}    & {\color[HTML]{303030} 0.678}          \\
\multicolumn{1}{c|}{\textbf{Beta(0.5,0.5)}}          & {\color[HTML]{303030} 0.754(79\%)}           & {\color[HTML]{303030} 0.75}           & {\color[HTML]{303030} \textbf{0.73}}   & {\color[HTML]{303030} \textbf{0.73}}     & {\color[HTML]{303030} 0.827}          \\
\multicolumn{1}{c|}{\textbf{Bernoulli(0.5)}}         & {\color[HTML]{303030} 0.506(100\%)}          & {\color[HTML]{303030} \textbf{0.50}}  & {\color[HTML]{303030} 0.577}           & {\color[HTML]{303030} \textbf{0.50}}     & {\color[HTML]{303030} 0.971}          \\
\multicolumn{1}{c|}{\textbf{50\% 0, 50\% 0.8}}       & {\color[HTML]{303030} \textbf{0.666(80\%)}}  & {\color[HTML]{303030} 0.751}          & {\color[HTML]{303030} 0.736}           & {\color[HTML]{303030} 0.679}             & {\color[HTML]{303030} 0.968}
\\
\\
\end{tabular}
\caption{
We see that ATGA performs well in many settings. If we focus
    on {\em provable} strategy-proof mechanisms, then TGA and the optimal multiple deadline mechanism also often perform better than the serial cost sharing mechanism.
    }
\end{table}

\section{Minimizing Worst-Case Max-Delay and Sum-Delay}
\label{sec:prior-free}

In this section, we focus on prior-free mechanism design, without assuming that we have the prior distribution over the agents' types.  For both Max-Delay and
Sum-Delay, the notion of optimal mechanism is not well-defined.  Given
two mechanisms $A$ and $B$, mechanism $A$ may outperform mechanism $B$ under
some type profiles, and vice versa for some other type profiles.

We adopt the following dominance relationships for comparing mechanisms.

\begin{definition}

    Mechanism $A$ \textsc{Max-Delay-Dominates} mechanism $B$, if
    \begin{itemize}

        \item for \emph{every} type profile, the Max-Delay under
            mechanism $A$ is \emph{at
            most}\footnote{\label{footnote:tie-breaking}Tie-breaking detail:
            given a type profile, if under $A$, the project is not funded (max delay
            is $1$), and under $B$, the project is funded (the max delay happens to
            be also $1$), then we interpret that the max delay under $A$ is
            \emph{not} at most that under $B$.} that under mechanism $B$.

        \item for \emph{some} type profiles, the Max-Delay under
            mechanism $A$ is \emph{less than} that under mechanism $B$.

    \end{itemize}

    A mechanism is \textsc{Max-Delay-Undominated}, if it is not dominated by
    any \emph{strategy-proof} and \emph{individually rational} mechanisms.

\end{definition}

\begin{definition}

    Mechanism $A$ \textsc{Sum-Delay-Dominates} mechanism $B$, if

    \begin{itemize}

        \item for \emph{every} type profile, the Sum-Delay under
            mechanism $A$ is \emph{at most} that under mechanism $B$.

        \item for \emph{some} type profiles, the Sum-Delay under mechanism
            $A$ is \emph{less than} that under mechanism $B$.

    \end{itemize}

    A mechanism is \textsc{Sum-Delay-Undominated}, if it is not dominated by
    any \emph{strategy-proof} and \emph{individually rational} mechanisms.

\end{definition}












We first describe the {\em fixed deadline mechanism}. This mechanism resembles
the single deadline mechanism, but the main difference is that,
when the deadline arrives, all agents (including the free riders) get to consume
the project (if they haven't already started consuming), {\em even if
there isn't enough payment collected to cover the project cost}.

\vspace{.1in}
\noindent\framebox{\parbox{\textwidth}{%
    \begin{center}
    \textbf{Fixed Deadline Mechanism}
    \end{center}

    \flushright{%
    Strategy-proofness: Yes

    Individual rationality: Yes

    {\bf Ex post budget balance: No}}

    \begin{itemize}

        \item Set a fixed deadline of $0\le t_C \le 1$. Under the mechanism,
            an agent's allocation time is at most $t_C$.

        \item Consider the following set: \[ K = \{ k\ \vert\ \textnormal{$k$
            values among the $v_i$ are at least $\frac{1}{kt_C}$}, 1\le k \le n
            \} \]

        \item If $K$ is empty, then the project is not built. Every agent's
            allocation time is $t_C$ and pays $0$.

        \item If $K$ is not empty, then the highest $k^* = \max K$ agents
            each pays $1/k^*$ and have their release time set to $0$. The other
            agents start consuming the project at time $t_C$ and each pays $0$.

    \end{itemize}
}}
\vspace{.1in}

The idea essentially is that we run the serial cost sharing mechanism
on the time interval $[0,t_C]$, and every agent receives the time interval $[t_C,1]$
\emph{for free}.   The mechanism remains strategy-proof and
individually rational.  Unfortunately, the mechanism is not ex post budget
balanced---even if the cost sharing failed (\emph{e.g.}, $K$ is empty), we
still need to release the project to the agents at time $t_C$ for free. If $t_C<1$,
we have to fund the public project without collecting back any payments.

The reason we describe the fixed deadline mechanism is because our final
mechanism uses it as a sub-component, and the way it is used fixes the budget
balance issue.








\begin{example}
    Let us consider the type profile $(0.9, 0.8, 0.26, 0.26)$. We run the
    fixed deadline mechanism
    using different $t_C$ values:

    \begin{itemize}

        \item If we set $t_C=0.9$, then agent $1$ and $2$ would start consuming
            the project
            at time $0$ and each pays $0.5$. Agent $3$ and $4$ pay nothing but they have to wait until time $0.9$.

        \item If we set $t_C=0.7$, then agent $1$ and $2$ would still start consuming
            the project
            at time $0$ and each pays $0.5$. Agent $3$ and $4$ pay nothing but they only
            need to wait until $0.7$. This is obviously better.

        \item If we set $t_C=0.5$, then all agents pay $0$ and only wait until
            $0.5$.  However, we run into budget issue in this scenario.

    \end{itemize}

We need $t_C$ to be small, in order to have shorter delays.  However, if $t_C$
    is too small, we have budget issues. The optimal $t_C$ value depends on the
    type profile.  For the above type profile, the optimal
    $t_C=\frac{0.5}{0.8}=0.625$.  When $t_C=0.625$, agent $2$ is still willing
    to pay $0.5$ for the time interval $[0,0.625]$ as $0.8\times 0.625=0.5$.
\end{example}

\begin{definition}

    Given a type profile $(v_1,v_2,\dots,v_n)$, consider the following set:

    \[ K(t_C) = \{ k\ \vert\ \textnormal{$k$ values among the $v_i$ are at least
    $\frac{1}{kt_C}$}, 1\le k \le n \} \]

    $t_C$ is between $0$ and $1$. As $t_C$ becomes smaller, the set $K(t_C)$
    also becomes smaller.  Let $t_C^*$ be the minimum value so that $K(t_C^*)$
    is not empty.  If such $t_C^*$ does not exist (\emph{e.g.}, $K(1)$ is
    empty), then we set $t_C^*=1$.

    $t_C^*$ is called the \textbf{optimal deadline} for this type profile.

\end{definition}

Instead of using a constant deadline, we may pick the optimal deadline {\em for
every given type profile}.

\vspace{.1in}
\noindent\framebox{\parbox{\textwidth}{%
    \begin{center}
    \textbf{Optimal Deadline Mechanism}
    \end{center}

    \flushright{%
    {\bf Strategy-proofness: No}

    Individual rationality: Yes

    Ex post budget balance: Yes}

    \begin{itemize}
        \item For every type profile, we calculate its optimal deadline.
        \item We run the fixed deadline mechanism using the optimal deadline.
    \end{itemize}
}}
\vspace{.1in}

The optimal deadline mechanism is ex post budget balanced. If we cannot find $k$ agents to pay
$1/k$ each for any $k$, then the optimal deadline is $1$ and the cost sharing
failed. In this case, we simply do not fund the project.

Unfortunately, we gained some and lost some. The ``optimal'' deadline
makes the mechanism not strategy-proof.

\begin{example} Let us re-consider the type profile $(0.9, 0.8, 0.26, 0.26)$.
    The optimal deadline for this type profile is $0.625$.  By reporting
    truthfully, agent $2$ receives the project at time $0$ and pays $0.5$.
    However, she can lower her type to $0.26$ (the optimal deadline is now
slightly below $1$). Agent $2$ still receives the project at time $0$ but only pays
$0.25$.  \end{example}

Despite not being strategy-proof,
we can prove one nice property of the optimal deadline mechanism.

\begin{theorem} The optimal deadline mechanism is both
\textsc{Max-Delay-Undominated} and \textsc{Sum-Delay-Undominated}.
\end{theorem}

\begin{proof}

    We first focus on \textsc{Max-Delay-Undominance}.  Let $M$ be a
    strategy-proof and individually rational mechanism that
    \textsc{Max-Delay-Dominates} the optimal deadline mechanism.
    We will prove by contradiction
    that such a mechanism does not exist.

    Let $(v_1,v_2,\dots,v_n)$ be an arbitrary type profile. Without loss of
    generality, we assume $v_1\ge v_2\ge\dots v_n$. We will show that $M$'s
    allocations and payments must be identical to that of the optimal deadline
    mechanism for
    this type profile. That is, $M$ must be identical to the optimal deadline
    mechanism, which
    results in a contradiction.

    We first consider type profiles under which the project is funded
    under the optimal deadline mechanism.
    We still denote the type profile under discussion by
    $(v_1,v_2,\dots,v_n)$. Let $k^*$ be the number of agents who participate in
    the cost sharing under the optimal deadline mechanism.

    We construct the following type profile:

    \begin{equation}\label{tp:trivial}
        (\underbrace{1/k^*, \dots,1/k^*}_{k^*},0,\dots,0)
    \end{equation}

    For the above type profile, under the optimal deadline mechanism, the first $k^*$ agents
    receive the project at time $0$ and each pays $1/k^*$. By dominance assumption
    (both \textsc{Max-Delay-Dominance} and \textsc{Sum-Delay-Dominance}), under
    $M$, the project must be funded. To collect $1$, the first $k^*$ agents must
    each pays $1/k^*$ and must receive the project at time $0$ due to individual
    rationality.

    Let us then construct a slightly modified type profile:

    \begin{equation}\label{tp:v1}
    (v_1,\underbrace{1/k^*, \dots,1/k^*}_{k^*-1},0,\dots,0)
    \end{equation}

    Since $v_1\ge 1/k^*$, under $M$, agent $1$ must still receive the project at
    time $0$ due to Myerson's monotonicity characterization.  Agent $1$'s payment must
    still be $1/k^*$. If the new payment is lower, then had agent $1$'s true
    type been $1/k^*$, it is beneficial to report $v_1$ instead. If the new
    payment is higher, then agent $1$ benefits by reporting $1/k^*$ instead.
    Agent $2$ to $k^*$ still pay $1/k^*$ and receive the project at time $0$ due to
    individual rationality.

    We repeat the above reasoning by constructing another slightly modified
    type profile:

    \begin{equation}\label{tp:v2}
    (v_1,v_2,\underbrace{1/k^*, \dots,1/k^*}_{k^*-2},0,\dots,0)
    \end{equation}

    Due to the monotonicity constraint, agent $2$ still pays $1/k^*$ and
    receives the project at time $0$. Had agent $1$ reported $1/k^*$, he would
    receive the project at time $0$ and pay $1/k^*$, so due to the monotonicity
    constraint, agent $1$ still pays $1/k^*$ and receives the project at time $0$
    under type profile~\eqref{tp:v2}. The rest of the agents must be
    responsible for the remaining $(k^*-2)/k^*$, so they still each pays
    $1/k^*$ and receives the project at time $0$.

    At the end, we can show that under $M$, for the following type profile, the
    first $k^*$ agents each pays $1/k^*$ and must receive the project at $0$.

    \begin{equation}\label{tp:vk}
        (v_1,v_2,\dots,v_{k^*},0,\dots,0)
    \end{equation}

    For the above type profile~\eqref{tp:vk}, there are $n-k^*$ agents
    reporting $0$s.  For such agents, their payments must be $0$ due to
    individual rationality. Since $M$ \textsc{Max-Delay-Dominates}\footnote{The
    claim remains true if we switch to \textsc{Sum-Delay-Dominance}.}
    the optimal deadline mechanism, these agents' allocation time must be at most
    $\frac{1}{k^*v_{k^*}}$, which is their allocation time under the optimal deadline
    mechanism.  We show that they cannot receive the
    project strictly earlier than $\frac{1}{k^*v_{k^*}}$ under $M$.

    Let us consider the following type profile:

    \begin{equation}\label{tp:vk1}
        (v_1,v_2,\dots,v_{k^*},\frac{k^*v_{k^*}}{k^*+1},\dots,0)
    \end{equation}

For type profile~\eqref{tp:vk1}, agent $k^*+1$ must receive the project at time
    $0$ and pay $1/(k^*+1)$.  She can actually benefit by reporting $0$
    instead, if under type profile~\eqref{tp:vk}, agents reporting $0$ receive
    the project earlier than $\frac{1}{k^*v_{k^*}}$ for free.

    Therefore, for type profile~\eqref{tp:vk}, all agents who report $0$ must
    receive the project at exactly $\frac{1}{k^*v_{k^*}}$. That is, for type
    profile~\eqref{tp:vk}, $M$ and the optimal deadline mechanism are equivalent.

    Now let us construct yet another modified type profile:

    \begin{equation}\label{tp:vk1true}
        (v_1,v_2,\dots,v_{k^*},v_{k^*+1},0,\dots,0)
    \end{equation}

    Here, we must have $v_{k^*+1}<\frac{k^*v_{k^*}}{k^*+1}$. Otherwise, under
    the original type profile, we would have more than $k^*$ agents who join
    the cost sharing. We assume under $M$, agent $k^*+1$ receives the project at
    time $t$ and pays $p$. $t$ is at most $\frac{1}{k^*v_{k^*}}$ due to the
    monotonicity constraint. We have

    \begin{align*}
     \textnormal{utility when the true type is $v_{k^*+1}$ and reporting truthfully} &=v_{k^*+1}(1-t)-p\\
     \textnormal{utility when the true type is $v_{k^*+1}$ and reporting $0$} &= v_{k^*+1}(1-\frac{1}{k^*v_{k^*}})
    \end{align*}

    Therefore,
    \begin{equation}\label{eq:oneside}
    v_{k^*+1}(1-t)-p\ge v_{k^*+1}(1-\frac{1}{k^*v_{k^*}})
    \end{equation}

    Had agent $k^*+1$'s type been $\frac{k^*v_{k^*}}{k^*+1}$, her utility for reporting
    her true type must be at least her utility when reporting $v_{k^*+1}$. That is,

    \begin{align*}
    \textnormal{utility when the true type is $\frac{k^*v_{k^*}}{k^*+1}$ and reporting truthfully} &= \frac{k^*v_{k^*}}{k^*+1}-\frac{1}{k^*+1} \\
        \textnormal{utility when the true type is $\frac{k^*v_{k^*}}{k^*+1}$ and reporting $v_{k^*+1}$} &= \frac{k^*v_{k^*}}{k^*+1}(1-t) -p
    \end{align*}

    That is,
    \begin{equation}\label{eq:twoside}
    \frac{k^*v_{k^*}}{k^*+1}-\frac{1}{k^*+1} \ge \frac{k^*v_{k^*}}{k^*+1}(1-t) -p
    \end{equation}

    Combine Equation~\eqref{eq:oneside}, Equation~\eqref{eq:twoside},
    $v_{k^*+1}<\frac{k^*v_{k^*}}{k^*+1}$, and $t\le\frac{1}{k^*v_{k^*}}$, we
    have $p=0$ and $t=\frac{1}{k^*v_{k^*}}$.  That is, under type
    profile~\eqref{tp:vk1true}, agent $k^*+1$'s allocation and payment remain
    the same whether she reports $0$ or $v_{k^*+1}$.

    Repeat the above steps, we can show that under the following arbitrary
    profile, agent $k^*+2$ to $n$'s allocation and payment also remain the same
    as when they report $0$.

    \begin{equation}\label{tp:vkn}
        (v_1,v_2,\dots,v_{k^*},v_{k^*+1},v_{k^*+2},\dots,v_n)
    \end{equation}

    That is, for type profiles where the project is funded under the optimal
    deadline mechanism, $M$ and the optimal deadline mechanism are equivalent.

    We then consider an arbitrary type profile for which the project is not funded
    under the optimal deadline mechanism.
    Due to the monotonicity constraint, an agent's
    utility never decreases when her type increases. If any agent $i$ receives
    the project at time $t$ that is strictly before $1$ and pays $p$, then due to
    the individual rationality constraint, we have that $v_i(1-t)-p\ge 0$.
    $v_i$ must be strictly below $1$, otherwise the project is funded under
    the optimal deadline mechanism.  Had agent $i$'s true type been higher but still below $1$
    (say, to $v_i+\epsilon$), her utility must be positive, because she can
    always report $v_i$ even when her true type is $v_i+\epsilon$.  But earlier
    we proved that had $v_i$'s true type been $1$, she would receive the project at
    time $0$ and pay $1$. Her utility is $0$ when her type is $1$.  This means
    her utility decreased if we change her true type from $v_i+\epsilon$ to
    $1$, which is a contradiction. That is, all agents must receive the project at
    time $1$ (and must pay $0$). Therefore, for an arbitrary type profile for
    which the project is not funded under the optimal deadline mechanism, $M$ still behaves the same as the optimal deadline mechanism.

    In the above proof, all places where we reference
    \textsc{Max-Delay-Dominance} can be changed to
    \textsc{Sum-Delay-Dominance}.
\end{proof}

The optimal deadline mechanism is both \textsc{Max-Delay-Undominated} and
\textsc{Sum-Delay-Undominated}, but it is not strategy-proof.  We now propose
our final mechanism in this section.
The new
mechanism is strategy-proof and its delay is within a constant factor of
the optimal deadline mechanism.\footnote{That is, we fixed the strategy-proofness issue at the
cost of longer delays, but it is within a constant factor.}

\vspace{.1in}
\noindent\framebox{\parbox{\textwidth}{%
    \begin{center}
    \textbf{Group-Based Optimal Deadline Mechanism}
    \end{center}

    \flushright{%
    Strategy-proofness: Yes

    Individual rationality: Yes

    Ex post budget balance: Yes}

    \begin{itemize}
        \item For agent $i$, we flip a fair coin to randomly assign her to either the left group or the right group.
        \item We calculate the optimal deadlines of both groups.
        \item We run the fixed deadline mechanism on both groups.
        \item The left group uses the optimal deadline from the right group and vice versa.
    \end{itemize}
}}
\vspace{.1in}

\begin{claim} The group-based optimal deadline mechanism
is strategy-proof, individually rational, and ex post budget balanced.
\end{claim}

\begin{proof}
    Every agent participates in a fixed deadline mechanism so strategy-proofness and
    individual rationality hold. Let $D_L$ and $D_R$ be the optimal deadlines
    of the left and right groups, respectively. If $D_L < D_R$, then the left
    group will definitely succeed in the cost sharing, because its optimal
    deadline is $D_L$ and now they face an extended deadline. The right group
    will definitely fail in the cost sharing, as they face a deadline that is
    earlier than the optimal one.  At the end, some agents in the left group
    pay and receive the project at $0$, and the remaining agents in the left group
    receive the project at time $D_R$ for free.  All agents in the right group
    receive the project at time $D_L$ for free.  If $D_L > D_R$, the reasoning is
    the same. If $D_L=D_R < 1$, then we simply tie-break in favour of the left
    group. If $D_L=D_R=1$, then potentially both groups fail in the cost
    sharing, in which case, we simply do not fund the project.
\end{proof}

\begin{definition} Mechanism $A$ is \textsc{$\alpha$-Max-Delay-Competitive}
    against mechanism $B$ if for every agent $i$, every type profile, we have
    that the max delay under $A$ is at most $\alpha$ times the max delay under
    $B$.

    \textsc{$\alpha$-Sum-Delay-Competitive} is defined similarly.
\end{definition}

\begin{theorem}\label{th:4}
    The group-based optimal deadline mechanism is \textsc{$4$-Max-Delay-Competitive} against the optimal deadline mechanism under
    one additional technical assumption:
    \begin{itemize}
        \item At least one agent does not participate in the cost sharing under the optimal deadline mechanism. I.e., there is at least one free rider.
    \end{itemize}
\end{theorem}

The assumption is needed only
because in the single case of everyone joining the cost sharing under
the optimal deadline mechanism, the max delay is 0. While under the group-based
optimal deadline mechanism, the max delay is
always greater than 0 so it would not be competitive in this one case only.
As our system would welcome as many agents as possible,
it is expected that there are always agents who don't value the project very much
so that they would prefer to be free riders instead of participating in the
cost sharing under the optimal deadline mechanism.
For example, if the use case is crowd funding online novel series, and paying readers get to
            read the new chapters slightly earlier (enjoying a premium period), then
            it is not unrealistic to assume that free riders exist.

\begin{proof} Let us consider an arbitrary type profile that satisfy
    the additional
    technical assumption. We denote it by $(v_1,v_2,\dots,v_n)$.  Without loss
    of generality, we assume $v_1\ge v_2\ge\dots\ge v_n$.  Let $k^*$ be the
    number of agents who join the cost sharing under the optimal deadline mechanism.  The
    optimal deadline under the optimal deadline mechanism is then $D^*=\frac{1}{k^*v_{k^*}}$,
    which is exactly the max delay for this type profile.

    Under a specific random grouping, for the set of agents from $1$ to $k^*$,
    we assume $k_L$ agents are assigned to the left group and $k_R=k^*-k_L$
    agents are assigned to the right group.

    For the left group, the optimal deadline is at most
    $\frac{1}{k_{L}v_{k^*}}$ if $k_L\ge 1$, which is at most
    $\frac{k^*}{k_{L}}D^*$.  When $k_L=0$, the optimal deadline is at most $1$.
    Under the optimal deadline mechanism, since all types are at most $1$, the optimal deadline
    $D^*$ is at least $1/k^*$.  That is, if $k_L=0$, the optimal deadline of
    the left group is at most $k^*D^*$.

    In summary, the optimal deadline of the left group is at most
    $\frac{k^*}{k_L}D^*$ if $k_L\ge 1$ and $k^*D^*$ if $k_L=0$.  That is, the
    optimal deadline of the left group is at most $\frac{k^*}{\max\{1,
    k_L\}}D^*$

    Similarly, the optimal deadline of the right group is at most
    $\frac{k^*}{\max\{1, k_R\}}D^*$

    The max delay under the group-based optimal deadline mechanism is at most the worse of these two
    deadlines.  The ratio between the max delay under the group-based optimal deadline mechanism and the
    max delay under the optimal deadline mechanism is then at most $\frac{k^*}{\max\{1,
    \min\{k_L, k^*-k_L\}\}}$.

    We use $\alpha(k)$ to denote the expected ratio (expectation with regard to
    the random groupings):

    \begin{equation}\label{eq:alphak}
    \alpha(k)=\sum_{k_L=0}^{k}\frac{1}{2^{k}}{k\choose k_L}\frac{k}{\max\{1, \min\{k_L, k-k_L\}\}}
    \end{equation}

    We define $\beta(k)=\alpha(k)2^k$.
    \begin{equation*}
            \beta(k)=\sum_{k_L=0}^{k}{k\choose k_L}\frac{k}{\max\{1, \min\{k_L, k-k_L\}\}}
            =\sum_{k_L=1}^{k-1}{k\choose k_L}\frac{k}{\min\{k_L, k-k_L\}} + 2k
    \end{equation*}
    If $k$ is even and at least $50$, then
        \begin{align*}
            \beta(k)
            &=
            \sum_{k_L=1}^{k/2-1}{k\choose k_L}\frac{k}{\min\{k_L, k-k_L\}}
            +\sum_{k_L=k/2+1}^{k-1}{k\choose k_L}\frac{k}{\min\{k_L, k-k_L\}}
            + 2{k \choose k/2} + 2k \\
            &=
            2\sum_{k_L=1}^{k/2-1}{k\choose k_L}\frac{k}{k_L}
            + 2{k \choose k/2}+ 2k \\
        \end{align*}
        \begin{align*}
            \beta(k)
            &=
            2\sum_{k_L=1}^{k/2-1}{k+1\choose k_L+1}\frac{(k_L+1)k}{(k+1)k_L}
            + 2{k \choose k/2}+ 2k \\
            &\le
            4\sum_{k_L=1}^{k/2-3}{k+1\choose k_L+1} +
            2{k+1\choose k/2-1}\frac{(k/2-1)k}{(k+1)(k/2-2)}\\
            &+
            2{k+1\choose k/2}\frac{(k/2)k}{(k+1)(k/2-1)} +
            2{k+1 \choose k/2}+ 2k \\
            &\le
            4\sum_{k_L=1}^{k/2-3}{k+1\choose k_L+1} +
            2.1{k+1\choose k/2-1} + 4.1{k+1\choose k/2}+ 2k
        \end{align*}
        The ratio between ${k+1\choose k/2}$ and ${k+1\choose k/2-1}$ is at
        most $1.08$ when $k$ is at least $50$.
        \begin{equation*}
\beta(k) \le 4\sum_{k_L=1}^{k/2-3}{k+1\choose k_L+1} + 4{k+1\choose k/2-1} + 4{k+1\choose k/2}+ 2k \le 4\sum_{k_L=0}^{k/2-1}{k+1\choose k_L+1} \le 4\times 2^k
        \end{equation*}




We omit the similar proof when $k$ is odd. In summary, we have $\alpha(k)\le 4$ when $k\ge 50$. For smaller $k$, we
numerically calculated $\alpha(k)$. All values are below $4$.
\end{proof}

\begin{corollary}
    The group-based optimal deadline mechanism is \textsc{$8$-Sum-Delay-Competitive} against the optimal deadline mechanism under one additional
    technical assumption:
    \begin{itemize}
        \item At least half of the agents do not participate in the cost sharing under the optimal deadline mechanism. I.e., the majority of the agents are free riders.
    \end{itemize}
\end{corollary}

\begin{proof}
    Let $D^*$ and $k^*$ be the optimal deadline and the number of agents who
    join the cost sharing under the optimal deadline mechanism.  The Sum-Delay of the
    agents under the optimal deadline mechanism is $(n-k^*)D^*$.  Under the group-based optimal deadline mechanism, the
    deadlines are at most $4D^*$ according to Theorem~\ref{th:4}.  The
    Sum-Delay is then at most $4D^*n$.  Therefore, the competitive
    ratio is $\frac{4n}{n-k^*}$, which is at least $8$ if $k^*\le n/2$.
\end{proof}

\section{Conclusion and Future Work}

In this paper, we studied a binary excludable public project model where the mechanism can specify when an agent can start consuming the project.
The waiting time faced by an agent is called her delay. Delay, which we aim to minimize, is only implemented to incentivize enough payments to cover the project cost.
We designed various mechanisms to minimize the maximum delay and the total delay among all agents.
For both expected delay objectives, we proved that a
near-optimal single deadline mechanism exists when there are a large number of
agents, irrespective of prior distributions. Besides, for a smaller number of
agents, we proposed an automated mechanism design approach based on evolutionary computation.
For prior-free settings, we proposed the group-based optimal deadline mechanism, which was shown to be competitive against an undominated mechanism.

As the problem setting is rather new, there are plenty of options to be
explored when designing mechanisms with better performance. Possible solutions
showing promise include, for exmaple, another method we considered but did not
dedicate as much time into---scheduling fixed prices for different sections of
time periods, regardless of the agents' submitted valuations. But such a
mechanism will require extensive simulations and analyses to evaluate its
performance.


For our results presented under the prior-free settings, we made a
certain number of assumptions, some of which easily hold true for realistic
applications---and therefore rather natural---some of which less so. For
example, there is an assumption that there is at least one agent who is
not willing to participate in the cost sharing under the optimal deadline
mechanism. This is necessary because otherwise the optimal deadline mechanism's
delay is 0 and no strategy-proof mechanisms can match it.
This assumption can also be easily satisfied
by including free riders who are determined not to contribute at all.
Either removing existing constraints to generalize
the solution or adding more assumptions to yield better results would be
reasonable as immediate future work.

\bibliographystyle{theapa}
\bibliography{/home/mingyu/nixos/newmg.bib}
\end{document}